\def\testend{ \end{document}}

\documentclass[numbook,onecollarge]{article}

\def\testend{\end{document}}
\usepackage{amsmath}
\usepackage{amsxtra}
\usepackage{amssymb}
\usepackage{amsthm}
\usepackage{cite}
\usepackage[arrow,curve,matrix,arc,2cell]{xy}
\usepackage{graphicx}
\usepackage[top=2cm, bottom=2.5cm, left=3cm, right=2cm]{geometry}
\theoremstyle{plain}
\newtheorem{theorem}{Theorem}[section]
\newtheorem{lemma}{Lemma}[theorem]
\newtheorem{proposition}[theorem]{Proposition}
\newtheorem{corollary}[theorem]{Corollary}

\newtheorem{definition}{Definition}[section]

\newtheorem{comment}{Comment}[section]
\def\vsk{ \vskip.075in\noindent }

\def\medskip{\vskip.15in\noindent}

\def\ket#1{|#1 \rangle}
\def\bra#1{\langle #1|}

\def\tilcal#1{\tilde{\cal C}}
\def\ncl{{\not {\hskip-.03in C}}}
\def\pr{^\prime}
\def\half{\textstyle{1\over 2}}

\def\sys{{\cal S}}
\def\cfield{\mathbb{C}}
\def\rfield{\mathbb{R}}
\def\qfield{\mathbb{H}}

\begin{document}
\title{{Derivation of the Rules of Quantum Mechanics from 
Information-Theoretic Axioms
}}
\author{Daniel I. Fivel $^\dagger$}
\date{\today}
\maketitle
\centerline{{\bf Abstract}}
\vsk
Conventional  quantum mechanics with a complex Hilbert space and the Born Rule is derived from five axioms describing
 properties of  probability distributions for the outcome of measurements.
Axioms I,II,III are common to quantum mechanics and  hidden variable theories.
 Axiom IV recognizes a
phenomenon, first noted  by Turing and von Neumann,  in which the increase in entropy resulting from a measurement is reduced by a suitable intermediate measurement.  This is shown to be impossible for local hidden variable theories.   Axiom IV, together with the first three, almost suffice to deduce the conventional  rules but allow some exotic, alternatives such as  real or quaternionic 
quantum mechanics.
 Axiom V recognizes a property of the distribution of outcomes of random measurements on  qubits  which holds only  in the complex Hilbert space model. It is then shown that the five axioms also imply the conventional rules  for
 all dimensions.  
\vsk
\section{Introduction}
Because the conventional rules of quantum mechanics (CRQM) are abstract and unintuitive, there have been many attempts since their formulation by Dirac and von Neumann to derive them from axioms extracted from the experimental data. {\em The axioms are supposed to isolate the essential ingredients of the data  that require the CRQM.}  Specifically they should answer the following questions: (1) Why must we
identify the set $\sys$ of pure states with a Hilbert space and in particular  a Hilbert space over the complex field
 ($\cfield$) rather than the real field ($\rfield$) or the  quaternions ($\qfield$)?  (2) 
Why must we calculate probabilities by the Born Rule? (3)  What precisely is it  that excludes a hidden-variable explanation?
\vsk
We obtain an important clue to the construction of a satisfactory axiomatic system that addresses these questions when we recall that the
celebrated quantum logic approach of Birkhoff and von Neumann\cite{BvN} gives only a very  limited restriction on the model. It endows the
system with a projective space structure and hence gives meaning to ``dimension"  but imposes
no restriction at all for dimension $N = 2$ (e.g.\  qubits or spin $\half$)  and almost none for  $N = 3$ (e.g.\ spin $1$). 
Even for $N > 3$, it only tells us that $\sys$ is a projective space over some sort of skew field $\mathbb{F}$, which is
insufficient to derive the Born Rule by means of Gleason's Theorem\cite{GLEA}. Even if $\mathbb{F}$ is assumed to be $\cfield$,  Gleason's theorem only works for $N > 2$. 
\vsk
\noindent
The fact that  the CRQM for qubits is the most difficult to derive suggests an axiomatic strategy based on information theory in which  $N = 2$ subspaces play the fundamental role.  
Reinforcement for this strategy comes from the recognition \cite{SYK,SPBC,WOOT} that,
even for a system prepared in a pure state, there are different probability distributions for the outcomes of different measurements.
  When this is taken into account, it was shown in  \cite{SPBC} that the  appropriate measure of the
removal of uncertainty when the outcome of a measurement is known
is not the von Neumann entropy, but rather the so-called  {\em information entropy} computed by averaging the Shannon entropy over the outcomes of all possible measurements.
Unlike the von Neumann entropy, the
\vsk$----------------------$\\
\noindent $\dagger$
Department of Physics, University of Maryland

College Park, MD 20742

Electronic address: fivel@umd.edu
\eject
\noindent
 information entropy is {\em model dependent}, and as the following table shows,  can distinguish  between real, complex, and quaternionic quantum mechanics in $N = 2$ subspaces.
\vsk
\centerline{Pure state entropy for $N=2$.}
\begin{equation}\label{purestateentropy}
\begin{matrix}
 &\rfield&\cfield&\qfield\\
 \text{von Neumann entropy}&0&0&0\\
\text{information entropy}&2\ln 2 - 1&1/2&7/12\\
\end{matrix}
\end{equation}
The calculations can be found  in section {\bf{\ref{pickingout}}} below.
\vsk
We shall exploit these observations to construct a derivation of the CRQM. 
\vsk
\section{Preliminary observations}
The data consists of a table of the function
\begin{equation}
 p(x,y) \equiv x(y) \to [0,1]
 \end{equation}
 which gives the attenuation of a beam of particles prepared by a device labeled $x$ after passage through a detecting device labeled $y$.  In optics, for example, $x$ labels the polarizers and $y$ the analyzers .  
  Since the same device may play either role, we  use the same set 
$\sys$  of labels for each argument. We  refer to devices labeled by the elements of $\sys$ as {\em filters} \cite{MIEL}. We refer to a table giving the values of $x(y)$ for
$\sys$ as  its  $p$-table. 
\vsk
We confine our discussion to  systems which the CRQM describes by  a Hilbert space of finite dimension $N$. With this
restriction the CRQM may be stated as follows:
\vsk 
{\bf The CRQM}:
There is  a map $z \to \bar{z}$  from  $\sys$ to a complex projective space  $ \cfield P^{N-1}$ for some finite $N$
 such that
\begin{equation}\label{bornrule}
x(y) =  |\hat{x}^*\cdot \hat{y}|^2
\end{equation}
where 
$$\hat{z} \equiv \bar{z}/|\bar{z}|. \nonumber$$
\noindent
We refer to this map as  the  {\em Born Rule correspondence} (BRC) between ${\cal S}$ and $\cfield P^{N-1}$.
\vsk
{\em Our axioms will   be  a minimal set of properties of a $p$-table from which we can deduce the CRQM as stated. Thus
our axiomatization  is based on the empirical data in contrast to that of Hardy\cite{HARDY} which is based on what
he describes as ``reasonable axioms which might well have been posited without any particular access to the empirical data".}
\vsk
One immediately observes two kinds of structure in the $p$-table of a quantum mechanical system: a metrical structure and a statistical structure. What is special about the $p$-table of a quantum mechanical system is the way these two structures interact.
\vsk
 It will be helpful
to have before us the $p$-table of a very simple 
system $\sys_o$ with $N = 2$  which exhibits the properties we will examine in this section.
 This table is a skeleton of quantum mechanics in which there are just
 four filters  $x,x\pr,y,y\pr$ corresponding to  states of linearly polarized light which the CRQM describes by
\begin{equation}\label{CRQMreal}
\hat{ x} = (1,0),\; \; \hat{x}\pr = (0,1), \; \; \hat{y}= 2^{-1/2} (1,1),\; \;  \hat{y}\pr=  2^{-1/2} (1,-1).
\end{equation}
The  $p$-table for $\sys_o$ is:
\begin{equation}\label{toymodel}
\begin{matrix}
 &x&x\pr&y&y\pr\\
x&1&0&1/2&1/2\\
x\pr&0&1&1/2&1/2\\
y&1/2&1/2&1&0\\
y\pr&1/2&1/2&0&1\\
\end{matrix}
\end{equation}
\vsk
The metrical structure is seen in the resemblance between the $p$-table and a  road atlas which
gives the  distances between a set of  cities. Our problem is analogous to that of determining the geometry of the earth
from such an atlas.
Like a road atlas, which has the special value zero on and only on the diagonal, the $p$-table has the special value unity on and only on the diagonal.  Like an atlas it is also symmetric about the diagonal. 
These properties of $\sys_o$ are shared by the $p$-table of any quantum mechanical system
and are expressed by our first axiom:

\vsk
{\bf Axiom I}:
For every  $w,z  \in {\cal S}$:
\begin{equation}\label{basic}
  w(z) = 1 \text{ if and only if } w = z,
\end{equation}
$$w(z) = z(w). \nonumber$$

\vsk
The conditional probability that a particle in a beam prepared by a $u$ filter
will pass a $w$ filter given that it has first passed a $v$ filter is $u(v)v(w)$. In a classical system the result
would be the same if $v$ and $w$ were interchanged.  It follows from (\ref{basic}) that this is the case {\em for all} $u$ if and only if  either $v(w) = 1$ or $v(w) = 0$. In the former case the the two elements are identical.
\begin{definition}
Two elements $u,v$ such that $u(v) = 0$ are said to be orthogonal and we write $u \perp v$ . If $u,v$ are either identical or orthogonal they are said to be ``compatible" or ``classically related".  If they are neither identical nor orthogonal we write $u \ncl v$. \end{definition}
\vsk
In $\sys_o$ observe that  $x\ncl y, x\ncl y\pr,x\pr \ncl y,$ and $x\pr \ncl y\pr$  
so that $\sys_o$ cannot be decomposed into
the union of two subsets in which the elements of one are classically related to all members of the other.  
\vsk
\begin{definition}
A {\em frame} in a subset $\sys^*$ of $\sys$ is a maximal set of mutually orthogonal elements. 
\end{definition}
\vsk
It follows from (\ref{basic}) that if $x(z) = y(z)$ for all $z$, then $x(y) = 1$, and hence $x = y$.
Thus there is a one-one correspondence between the elements $x$ and the functions $x(\; )$, and
there is a natural metric on $\sys$:
\begin{equation}\label{dmetric}
d(x,y) = \sup_{z \in \sys}|x(z) - y(z)|.
\end{equation}
\begin{proposition}
The functions $u(\;)$  are continuous  with respect to the topology defined by the $d$-metric.
\end{proposition}\noindent
\begin{proof}
\begin{equation}
|u(v) - u(w)| = |v(u) - w(u)| \leq \sup_{u\in \sys}|v(u) - w(u)|  = d(v,w).
\end{equation}
\end{proof}
\vsk
{\bf Note:}  Since there is a one-one correspondence  $u \leftrightarrow u(\;)$, we shall often
use $u$ to mean $u(\;)$ when no confusion will result.
\vsk
Observe that  $d(u,v)$  takes values between zero and unity, the former if and only if  $u$ and $v$  are identical and the latter if and only if $u$ and $v$ are orthogonal.
Thus the elements of a frame are a maximal set of maximally separated elements in the $d$-metric.
\vsk
The following axiom insures that the number of elements in any frame is finite and that ${\cal S}$ is complete.
\vsk
{\bf Axiom II: } 
{\em ${\cal S}$ is compact in the $d$-metric.}
\vsk
\begin{proposition}
The number of elements in any frame is finite.
\end{proposition}
\begin{proof}
A compact metric space is closed and totally bounded \cite{SIER}.  The latter implies that at most a finite number of elements
can be separated by any constant amount.
\end{proof}
\noindent
{\bf Note}:  It is this axiom  that will  restrict us to systems represented in
the CRQM   by   finite dimensional Hilbert spaces. The closure property will play a role
later in the discussion.

\vsk
A frame consisting of the elements ${\bf x} = \{x_1,x_2,\cdots\}$ will be denoted $F_{\bf x}$. In $\sys_o$ there are just two frames $F_{\bf x} = \{x,x\pr\}$ and $F_{\bf y} = \{y,y\pr\}$.
\vsk
\begin{definition}
A {\em frame function} on 
 $\sys^*$ is a probabiity distribution function on every frame  in $\sys^*$, i.e.\ a  function $\zeta$ from $\sys^*$ to $[0,1]$ such that 
 \begin{equation}
\sum_j \zeta(x_j) = 1 \;\text{ for every frame } F_{\bf x}.
\end{equation}
\end{definition}
\vsk
Note:  We use the term ``frame function" for what Gleason\cite{GLEA} calls a non-negative frame function of weight unity.
\vsk
{\bf Axiom III: } 
{\em The function $a(x) $ is a frame function on $\sys$ for every $a \in \sys$.}
\vsk
We refer to frame functions formed in this way from the elements of $\sys$ as {\em pure states}.
\vsk
The elements of a frame are the possible outcomes of some measurement. The physical interpretation of Axiom III is that   $a(x_j)$ is the probability of the  outcome $x_j$ when the measurement is performed on a system in
the state $a$. In $\sys_o$ there are four pure states  $x,x\pr,y,y\pr$ .
\vsk
In view of the symmetry of $p$, Axiom III also says that for every frame $F_{\bf x}$ in $\sys$ we have
\begin{equation}\label{frameisbasis}
\sum_{j} x_j(a) = 1 \text{ for every element } a \in \sys.
\end{equation}
A set of functions with this property is said to be a {\em basis} of $\sys$, whence  Axiom III is equivalent to the assertion
 that every frame in $\sys$ is a basis of $\sys$.
 \vsk  We are led to the notion of ``subspace" by considering  subsets $\sys^*$ of $\sys$ for which
 a given ({\em not necessarily maximal}) set of mutually orthogonal elements is a basis. For elements in this subset, the only  possible outcomes of a specific measurement 
will be one of the finite set of elements of the basis.  
\begin{definition}
Let $ F_{\bf x}^* \equiv \{x_1,\cdots,x_N\}$ be a  set of mutually orthogonal elements of $\sys$.
The subspace $\,\sys^*$  of $\sys$ spanned by $F_{\bf x}^*$ consists of elements $z$  for
which $F_{\bf x}^*$ is a basis i.e.\ for which
\begin{equation}\label{subspace}
\sum_{j=1}^N x_j(z) = 1.
\end{equation}
\end{definition}
\noindent
Evidently $\{x_1,\cdots,x_N\}$ is a maximal orthogonal set in $\sys^*$ since no element can be orthogonal
to all of them and satisfy (\ref{subspace}). Hence $F_{\bf x}^* \equiv \{x_1,\cdots,x_N\}$ is a frame in $\sys^*$. 
\vsk
\begin{proposition}\label{anyorthog}
An element orthogonal to every element of a frame in $\sys^*$ is orthogonal to every element of $\sys^*$.
\end{proposition}
\begin{proof}
If $y$ is orthogonal to a frame $F_x = \{x_1,\cdots,x_N\}$  on $\sys^*$, then $\{y,x_1,\cdots,x_N\}$ is a subset of some maximal set of mutually orthogonal elements of $\sys$ . Hence by (\ref{frameisbasis}) $\sum_{j=1}^N {x_j}(z) + y(z) \leq 1$. But the first term is unity for $z \in \sys^*$ whence
$y(z) = 0$.
\end{proof}
\begin{proposition}\label{anyframespans}Every frame in $\sys^*$ has  the {\em same number}  of elements and spans $\sys^*$.
Moreover, if $\{x_1,x_2,\cdots,x_N\}$ and $\{y_1,y_2,\cdots,y_N\}$ are two such frames, then
\begin{equation}\label{allframesequivalent}
\sum_{j=1}^N {y_j}(w) = \sum_{j=1}^N {x_j}(w) \text{ for all } w \in \sys
\end{equation}
with common value unity if and only if $w \in \sys^*$. 
\end{proposition}
\begin{proof}
Extend $F_{\bf x}^*$ to a frame $F_{\bf x} = \{x_1,x_2,\cdots, x_N,x_{N+1}, \cdots \}$ in $\sys$.
By proposition {\bf{\ref{anyorthog}}}  the elements $\{x_{N+1},\cdots\}$ are orthogonal to every
element of $\sys^*$ and hence to any maximal set $\{y_1,y_2,\cdots \}$ of mutually orthogonal
elements of $\sys^*$. If an element $z$ is orthogonal to $x_{N + 1},\cdots$, then $\sum_{j=1}^N {x_j}(z) = 1$,
so $z \in \sys^*$ and hence cannot be orthogonal to $\{y_1,y_2,\cdots\}$ since this is a maximal set of
mutually orthogonal elements of $\sys^*$.  Hence $\{y_1,y_2,\cdots,x_{N+1},\cdots\}$ is a maximal
set of mutually orthogonal elements of $\sys$, so that for all $w \in \sys$ we have
\begin{equation}
\sum_j {y_j}(w) + \sum_{j= N+1}^\infty {x_j}(w) = 1.
\end{equation}
Since 
$\sum_{j = 1}^N {x_j}(w) + \sum_{j=N+1}^\infty {x_j}(w) = 1$ it follows that
$\sum_j {y_j}(w) = \sum_{j = 1}^N {x_j}(w)$ with common value unity if and only if $w \in \sys^*$.
Thus $F_{\bf y}^* = \{y_1,y_2,\cdots\}$ spans $\sys^*,$ and it remains to prove that this set
contains $N$ elements.
Let $ \{x_1,\cdots,x_N\}$ and $ \{y_1,\cdots,y_{M}\}$ be frames on $\sys^*$ so that for any $z,w \in \sys^*$ 
\begin{equation}
\sum_{i=1}^N {x_i}(z) = 1 = \sum_{j=1}^ {M}{y_j}(w).
\end{equation}
Put $z = y_j$ on the left side and sum on $j$ to obtain $M$, and put $w = x_i$ on the right side and sum on $i$ to obtain $N$.  The two double sums are equal by the symmetry (\ref{basic}). 
\end{proof}
\begin{definition}
The number $N$ is referred to as the {\em dimension} of the subspace $\sys^*$.
\end{definition}
\noindent
\begin{corollary}\label{uniqueantipode}
For every element $x$ of a two-dimensional subspace, there is  a unique element $x\pr$ to which it is orthogonal.
\end{corollary}
\begin{proof}
If $x \perp x\pr$ and $x \perp x''$, then $x(x'') + {x\pr}(x'') = 1 \implies {x\pr}(x'') = 1$, i.e.\ $x\pr = x''$ by (\ref{basic}).
\end{proof}
\noindent
We refer to $x$ and $x\pr$ as {\em antipodes} and indicate the two dimensional space that they span by
 ${\cal P}_{xx\pr}$.
\vsk
\section{Effect of measurements on states.}
The frame function property of the elements $a \in \sys$ means that the probability  $b_j(a)$ sums to unity for a particle in a beam prepared by an $a$-filter to pass one of the filters $b_1,b_2,\cdots$ of a frame $F_{\bf b}$. 
The fact that $b(b) = 1$ for any $b$  means that a beam which is unattenuated by a $b$ filter will be unattenuated by a second $b$ filter.  Thus a determination of which  filter of frame $F_{\bf b}$  a system in state $a$ passes transforms
the probability distribution function  $a$ into the probability distribution function
 $F_{\bf b} a$ given by:
\begin{equation}\label{mixedstate}
a \to F_{\bf b} a = \sum_j b_j(a){b_j}.
\end{equation}
We shall refer to this transformation as a {\em measurement} of $a$ by $F_{\bf b}$.
It is the transformation  referred to in quantum mechanics as ``collapse of the wave function".
\vsk
In view of Axiom III,  $F_{\bf b} a$ is a frame function, since
 if $\rho_1,\rho_2,\cdots $ are frame functions so also is any convex linear combination of them i.e.\  any function  $\zeta = \sum_j \alpha_j \rho_j$ with $0 \leq \alpha_j \leq 1$ and $\sum_j \alpha_j = 1$. Frame functions such as $F_{\bf b} a$ that result  from measurements are of a special type
in that they are composed of {\em orthogonal} pure states. We refer to these as
 {\em mixed states}. The term ``state" will be used to refer to either a pure state or a mixed state frame function. 
 \vsk
  Note:  {\em It is important to keep in mind that we cannot assume at this stage  that frame functions composed  of convex combinations of  non-orthogonal pure states can be written as convex combinations of orthogonal pure states  --- something  we know to be a consequence in the CRQM  of the diagonalizability of convex combinations of  projection operators representing pure states. }
  \vsk
 The mixed state  $ \rho =\sum_j \alpha_j b_j$ 
 has the property that it is unchanged by a measurement of $F_{\bf b}$, i.e.\
 \begin{equation}
 F_{\bf b}\rho = F_{\bf b}\sum_j \alpha_j b_j = \sum_j \alpha_j  F_{\bf b}b_j = \sum_j \alpha_j b_j =\rho
 \end{equation}
 but will change for other choices of the frame.
 \vsk
 The Shannon entropy
 of a probability distribution $(q_1,q_2,\cdots)$
 \begin{equation}
 S[(q_1,q_2,\cdots)] = - \sum_j q_j \ln q_j.
 \end{equation}
  measures the removal of uncertainty when an outcome is known.
 As noted above,  Stotland et al  \cite{SPBC} observed that if we do not know in advance what measurement is going to be made on a state $\rho$, then to compute the  removal of uncertainty when a measurement is made on it
 we should take the average of the Shannon entropy of the outcome distributions  from all possible meaurements.
 This average is called the ``information entropy"  $\overline{S[\rho]}$. In contrast, the von Neumann entropy $H[\rho]$ considers only the
 distribution of outcomes when the measurement does not change the state, i.e.\  
 \begin{equation}
 H[\sum_j \alpha_j b_j] = - \sum_j \alpha_j \ln \alpha_j.
 \end{equation}
\vsk
Both the von Neumann entropy and the information entropy provide a measure of the purity of a mixed state.  Let us
compare them for states belonging to the skeletal system $\sys_o$ given by the $p$-table  (\ref{toymodel}):
\vsk
In $\sys_o$ there are just two frames $F_{\bf x}$ and $F_{\bf y}$. We have $F_{\bf x} x = x$ and $F_{\bf y} y = y$. There  appear to be
two mixed states $F_{\bf x} y = \half x + \half {x\pr}$ and $ F_{\bf y} x = \half y + \half {y\pr}$, but
these are in fact identical. To see this observe that 
$\half x(z) + \half {x\pr}(z) = \half =  \half (z(x) + z(x\pr)) $ for every $z$ and the same with
$y$ replacing $x$. The von Neumann entropy of the pure states is zero and of the mixed state is $\ln 2$.  The
information entropy is $\half \ln 2$ for the pure states and  $\ln 2$  for the mixed state.
\vsk
{\em In $\sys_o$ we see a rudimentary form of interference i.e.\  it can happen that two states formed in different ways (in this case the states $F_{\bf x} y$ and $F_{\bf y} x$)  turn out to be identical.} More generally it follows  from (\ref{allframesequivalent})  that if $F_{\bf x}^*$ and $F_{\bf y}^*$ are frames in the same $N$ dimensional subspace, then the maximally mixed states in that subspace are identical, i.e.\
\begin{equation}\label{maximallymixed}
  {1\over N}\sum_{j=1}^N {x_j} =  {1\over N}\sum_{j=1}^N {y_j}.
  \end{equation}
\vsk
Although $\sys_o$ exhibits this indistinguishability of maximally mixed states, {\em it is not a quantum mechanical system}. This
 may be seen by observing that
its statistical behavior can be produced by the following hidden variable model: Let $\xi$ be a hidden
variable that takes on values on a circle.  When $\xi$ is on the right and left halves of the circle, a particle passes filters $x,x\pr$  respectively,  and when $\xi$ is on the upper and lower halves of the circle, it passes  $y,y\pr$  respectively. The conditional probability
that $\xi$ be  on the right half, given that it is in the upper half is $1/2$. Similarly all of the values  in the table are predicted correctly.
\vsk
The $\sys_o$ model can be generalized to one with  sets of $2(R + 1)$ elements $x_1,x_2,\cdots,x_{R+1}$ and $x_1\pr,x_2\pr,\cdots,x_{R+1}\pr$
 such that $x_j(x_k) = \half$ and $x_j(x_k\pr) = \half$ for $j \neq k$.
 Just as $\sys_o$  which has $R = 1$ is the skeleton (\ref{CRQMreal}) of two dimensional quantum mechanics over the real number field, the  case $R = 2$ is a skeleton of two-dimensional complex quantum mechanics
with  $x_1,x_2,x_3$ corresponding to $(1,0),2^{-1/2}(1,1),2^{-1/2}(1,i)$ respectively and $x_1\pr,x_2\pr,x_3\pr$ corresponding to
$(0,1),2^{-1/2}(1,-1),2^{-1/2}(1,-i)$ respectively.
The case $R = 4$ is a skeleton of quaternionic
 quantum mechanics with  additional elements 
 $x_4,x_5$ and $x_4\pr,x_5\pr$  corresponding to  $2^{-1/2}(1,\pm j)$ and $2^{-1/2}(1,\pm k)$ where $1,i,j,k$ are the quaternionic units.
  None of these skeletons are quantum mechanical since, as for $\sys_o$,  the function $x(y)$ can be reproduced
by a hidden-variable  $\xi$ ranging over an $R$-sphere: If the
 cartesian coordinates of  $\xi$ are $(\xi_1,\cdots,\xi_{R+1})$, a particle passes  an $x_j$ or $x\pr_j$ filter according as  $\xi_j$ is positive or negative. 
 \vsk
 \section{The entropic Turing-von Neumann effect}
We see then that there is a crucial ingredient of quantum mechanics not accounted for in  $\sys_o$ or any of these generalizations. 
A clue to the missing ingredient comes from a comparison of the information entropy of a pure state in the three skeletal
models with that of real, complex, and quaternionic quantum mechanics given in the following table:
\vsk
\centerline{Information entropy for $N=2$ pure states}
\begin{equation}\label{compareinfoentropy}
\begin{matrix}
 &\rfield&\cfield&\qfield\\
\text{skeleton}&\half \ln 2&{2\over 3}\ln 2&{4\over 5}\ln 2\\
\text{qm}&2\ln2 - 1&1/2&7/12\\
\text{qm - skeleton}&0.040&0.037&0.029\\
\end{matrix}
\end{equation}
We see that the quantum mechanical values are slightly larger than those of the skeletons, which
tells us that there have to be more possible measurement frames to account for it. 
The effect of a measurement described by the frames $F_{\bf x}$ and $F_{\bf y}$ in $\sys_o$ is either to leave 
a pure state alone or  transform it into the maximally mixed state. Thus the von Neumann entropy change is either zero or $\ln 2$ in any measurement. {\em  In quantum mechanical systems, however, we observe that the increase in von Neumann entropy resulting from a measurement can be reduced by making a suitable intermediate measurement. } This phenomenon, which was first noted by
Turing\cite{TUR} and elaborated by von Neumann\cite{TvN} is related to the so-called Zeno effect in which the
dynamic evolution of a state is arrested by continuous monitoring.
\vsk
In order for the increase in von Neumann entropy resulting from the measurement of $F_{\bf b}$ on $a$ to
be reduced by some intermediate meaasurement, there must be a frame $F_{\bf c}$ such that 
\begin{equation}\label{tvn0}
H(F_{\bf b}F_{\bf c} a) < H(F_{\bf b} a).
\end{equation}
For $F_{\bf c}$ to accomplish this, it must transform $a$ into a state which in some sense lies ``between" $a$ and $b$. A natural
candidate for such a state is one of the form $\rho = \alpha a + (1 - \alpha) b$ with $0 < \alpha < 1$.  Although this is not a combination of
orthogonal pure states, we have seen above that {\em it is possible for two frame functions formed in different ways to be identical}, and hence  some measurement of $F_{\bf c}$ on $a$ can in principle produce a frame function identical to $\rho$.
 The resulting entropy will be smallest if the combination involves
just a single pair of orthogonal states $c,{c\pr}$. The following axiom asserts that such a measurement does indeed exist,
and, as we shall see, this will 
produce the entropic Turing-von Neumann effect. {\em It is the key property distinguishing quantum mechanical from classical systems.}
\vsk
{\bf Axiom IV: } 
{\em Given any two distinct elements  $a,b \in \sys$ there exists a pair of orthogonal elements $c,c\pr$ and a number $0 < \alpha < 1$ such that}
\begin{equation}\label{tvnproperty}
a(c) c + a(c\pr){c\pr} = \alpha a + (1 - \alpha)b  .
\end{equation} 
\noindent
The following are consequences of Axiom IV and Axiom I:
\vsk
\begin{proposition}\label{consequences}
\begin{equation}\label{consone}
 a(c) + a(c\pr) = 1.
\end{equation}
\begin{equation}\label{constwo}
 b(c) = a(c), \;\; b(c\pr) = a(c\pr).
 \end{equation}
 \begin{equation}\label{consthree}
 \alpha = \half.
 \end{equation}
 \begin{equation}\label{consfour}
 a(c) = \half(1 \pm \sqrt{a(b)}).
 \end{equation}\end{proposition} 
 \begin{proof}
 The left side of (\ref{tvnproperty}) is a frame function since the  right side is.
Hence equation (\ref{consone}) holds.  Substitution of $z = c$  in the arguments
of the functions $c,c\pr,a,b$ in (\ref{tvnproperty}) gives
\begin{equation}
a(c) = \alpha\, a(c) + (1 - \alpha)b(c),
\end{equation}
whence $a(c) = b(c)$  since $\alpha \neq 1$. Similarly substitution
of $z = c\pr$ gives $a(c\pr) = b(c\pr)$ proving (\ref{constwo}).
Substituting $z = a$ and $z = b$ in the arguments and
using (\ref{consone},\ref{constwo}) we have
\begin{equation}
a(c)^2 + (1 - a(c))^2 = \alpha  + (1-\alpha)b(a),
\end{equation}
$$
a(c)^2 + (1 - a(c))^2 = \alpha\, a(b) + (1 - \alpha,
$$
whence, since $a(b) \neq 1$ for $a \neq b$, we obtain (\ref{consthree})
and (\ref{consfour}).
\end{proof}
\vsk
 Since the interchange of $c$ and $c\pr$ changes $a(c)$ to $a(c\pr) = 1 - a(c)$ we can choose $c$ and $c\pr$
such that $a(c) \geq \half$, i.e.\  assume the solution with the $+$ in (\ref{consfour}).
\vsk
Hence we have:
\vsk
{\bf Equivalent form for Axiom IV}  
\vsk
{\em Given $a \neq b$ there exist orthogonal states $c,c\pr$ 
and a number $\lambda$ such that  
\begin{equation}\label{tvnproperty1}\half(a + b) = \lambda c + (1 - \lambda)c\pr\;\text{ with }
\half \leq \lambda \leq 1 .\hskip.5in 
\end{equation}}
It will then follow from the substitutions used above to obtain (\ref{consone} - \ref{consfour}) that
\begin{equation}\label{lambdaeq}
 \lambda = a(c) = b(c) =   \half(1 + \sqrt{a(b)}).
 \end{equation}
 \vsk
 {\em  Note that (\ref{tvnproperty1}) is a second example of an interference phenomenon in which two frame functions formed in different ways turn out to be identical.} What it says is that every mixture of two orthogonal states is indistinguishable from some equal mixture of
 two states. (Its validity in the CRQM follows from the diagonalizability of density matrices.)
\vsk
  We shall often use the special form taken by (\ref{tvnproperty1}) when the argument of the functions $a,b,c,c\pr$ is a member of the subspace $ {\cal P}_{cc\pr}$ spanned by $c,c\pr$ so that ${c\pr}(z) = 1 - {c}(z)$, namely:
\begin{equation}\label{tvnproperty2}
 \half(a(z) + b(z)) = (2 \lambda - 1) c(z) + (1 - \lambda), \;\; z \in {\cal P}_{cc\pr}.
 \end{equation}
 \noindent
  From (\ref{consone},\ref{constwo}) this applies to  $z = a,b$ and of course to $z = c,c\pr$.
 \vsk
We next derive a number of propositions that follow from Axioms I-IV that will play a role in the subsequent analysis:
\vsk
 \begin{proposition}\label{cunique}
 If $a$ and $b$ are  not orthogonal the solution $c,c\pr$ of  (\ref{tvnproperty1}) is unique.
 \end{proposition}
 \begin{proof}
 If $a,b$ are not orthogonal,  then $\lambda \neq \half$. 
If $g,g\pr$ are also solutions of (\ref{tvnproperty1}) it follows from 
 (\ref{tvnproperty2}) with  $z = g$ that
$ (2 \lambda - 1) c(g) =  (2 \lambda - 1) g(g) $ from which $c(g) = 1$ and hence $g = c$ .
\end{proof}
\vsk
{\bf Note}:  When we wish to emphasize that $c,{c\pr}$ and $\lambda$  in (\ref{tvnproperty1}) are uniquely determined by a non-orthogonal pair $a,b$, we will write
$c = c(a,b), c\pr = c\pr(a,b)$ and $\lambda = \lambda(a,b)$.
  \begin{proposition}\label{uniquesubspace}
 A pair of distinct states $a$, $b$ belong to a unique two dimensional subspace ${\cal P}_{ab}$.  
 \end{proposition}
 \begin{proof}
 It follows from (\ref{consone},\ref{constwo}) that $a,b$ belong to ${\cal P}_{cc\pr}$
 Suppose $a$ and $b$ belong to another two dimensional subspace ${\cal P}_{gg\pr}$.
 Add the two equations obtained from (\ref{tvnproperty1}) by setting $z = g$ and $z = g\pr$.
 Since $a,b \in {\cal P}_{gg\pr}$, the left sides will sum to unity, and we obtain
 \begin{equation}
1 = \lambda(c(g) + {c\pr}(g)) + (1- \lambda)({c}(g\pr) + {c\pr}(g\pr)).
\end{equation}
Since there is some maximal set of mutually orthogonal elements containing  $c,c\pr$, the coefficients of $\lambda$ and $(1 - \lambda)$ lie between
zero and unity as does $\lambda$.  Hence the only solution is for both coefficients to be unity. Hence $g,g\pr$ must lie
in the subspace ${\cal P}_{cc\pr}$, and hence by proposition {\bf{\ref{anyframespans}}} the subspaces ${\cal P}_{gg\pr}$ and
${\cal P}_{cc\pr}$ are identical. 
\end{proof}
\vsk
\begin{corollary}\label{onepointincommon}
Two two-dimensional subspaces are either disjoint, identical, or have just one point in common.
\end{corollary}
\begin{proposition}\label{anyorthogonal}
If $x$ is orthogonal to two distinct states $a,b$, it is orthogonal to every member of ${\cal P}_{ab}$.
\end{proposition}
\begin{proof}
From (\ref{tvnproperty1}) if $x \perp a$ and $x \perp b$, then $x \perp c$ and $x \perp {c\pr}$.  Hence
by proposition {\bf \ref{anyorthog}} $x$ is orthogonal to the subspace spanned by $c$ and ${c\pr}$ i.e.\
it is orthogonal to ${\cal P}_{ab}$.
\end{proof}
\begin{definition}
Two elements  $x,y$ are said to be ``mutually equatorial" if  $x(y) = \half$, and we write $x \smile y$. The set of elements
equatorial to a given element $z$ is referred to as the {\em equator opposite} $z$ and denoted ${\cal E}_z$.
\end{definition}
\begin{proposition}\label{equatorialtheorem}
If $a \ncl b$ belong to a two-dimensional subspace ${\cal P}$,  and $a,b \in {\cal E}_z$ for some $z \in {\cal P}$, then
$c(a,b) \in {\cal E}(z)$.
\end{proposition}
\begin{proof}
 Solve (\ref{tvnproperty2}) for $c(z)$ with $a(z) = b(z) = \half$.
  Noting that $\lambda \neq \half$ for $a \ncl b$  , we obtain $c(z) = \half. $
\end{proof}
\begin{proposition}
Let $a$ and $b$ be distinct and non-orthogonal, and let
 $b\pr$ be the antipode of $b$ in ${\cal P}_{ab}$. Then
\begin{equation}\label{eiscabprime}
c(a,b) \smile c(a,b\pr).
\end{equation}
\end{proposition}
\begin{proof}
Since $a(b) \neq 0$ or $1$ implies that $a(b\pr) \neq 0$ or $1$, it follows from
(\ref{tvnproperty1}) that there exists a unique  orthogonal pair  $e,e\pr$  such that
\begin{equation}\label{egammaequation}
\half (a(z) + b\pr(z)) = \gamma e(z) + (1 - \gamma){e\pr}(z),
\end{equation}
\noindent
where
\begin{equation}
e = c(a,b\pr), \; e\pr = c\pr(a,b\pr),\;\; \gamma = \lambda(a,b\pr) = \half(1 + \sqrt{p(a,b\pr)}).
\end{equation}
Since $a$ and $b\pr$ belong to ${\cal P}_{ee\pr}$ as well as ${\cal P}_{ab}$, it follows from proposition {\bf{\ref{onepointincommon}}} that $e,{e\pr}$ belong to ${\cal P}_{ab}$. Thus both $e,{e\pr}$ and ${b},b\pr$
span ${\cal P}_{ab}$, whence from proposition {\bf{\ref{anyframespans}}}
\begin{equation}
b(z) + b\pr(z) = e(z) + {e\pr}(z) \;\; \text{ for all } z \in \sys.
\end{equation}
Combining this with (\ref{egammaequation}), we obtain
\begin{equation}\label{differenceequation}
a(z) - {b}(z) = (2 \gamma - 1)( e(z) - {e\pr}(z)).
\end{equation}
Since $a(b\pr) \neq 0$, it follows that $\gamma \neq \half$. Putting $z = c(a,b)$ the left side is zero by (\ref{tvnproperty1}), whence $e(c) = e\pr(c)$.
But $e(c) + e\pr(c) = 1$ since $c \in {\cal P}_{ee\pr}$, whence $e(c) = \half$.
\end{proof}
\vsk
Since Axiom IV was motivated by the entropic Turing-von Neumann effect  we must verify that it does indeed follow from it, i.e.\
we must show that there is an intermediate
measurement that reduces the increase of von Neumann
entropy caused by a  measurement:
\begin{lemma}\label{monotonic}
Let $\rho = \alpha a + (1-\alpha)a\pr$ be a state in the two dimensional subspace ${\cal P}$, and let $\sigma$ be the maximally mixed state.  The von Neumann entropy $H(\rho)$ increases monotonically as the distance $d(\rho,\sigma)$ decreases. 
\end{lemma}
\begin{proof}
By (\ref{maximallymixed}) $\sigma = \half a + \half a\pr$ so that 
 \begin{equation}\label{drhosigma}
 d(\rho,\sigma) = \sup_{z}|\rho(z) - \sigma(z)| = |\alpha - \half| \sup_{z}|a(z) - a\pr(z)| =  |\alpha - \half|.
 \end{equation}
 The assertion follows from the fact that
 \begin{equation} \label{shannonf}
f(\mu) = -(\mu \ln \mu + (1-\mu)\ln(1-\mu))
\end{equation}
is a monotonically decreasing function of $|\mu - \half| $ on the interval $0 \leq \mu \leq 1$.
 \end{proof}
\begin{proposition}\label{tvNtheorem}
If  $a \neq b$, there is a measurement $F_x^* = \{x,x\pr\}$ with $x \in {\cal P}_{ab}$ such that
\begin{equation}
H(F_b F_x a) < H(F_b a).
\end{equation}
\end{proposition}
\begin{proof} 
Let  $\sigma$ be the maximally mixed state in ${\cal P}_{ab}$.
Let $c = c(a,b)$ and $e = c(a,b\pr)$ as defined in the note following Proposition  {\bf{\ref{cunique}}}.  If $a(b) \geq \half$, choose $x = c$, and if
$a(b) < \half$, choose $x = e$. We have:
$F_b a = \tau  b + (1 - \tau)b\pr$ with $\tau = a(b)$, whence
$d(F_b a,\sigma) = |\tau - \half| = |a(b)- \half|$. On the other hand
 $F_b F_c a = F_b(\half( a + b)) = \mu  b + (1-\mu)b\pr$ with $\mu = \half(1 + a(b))$, whence
$d(F_b F_c a,\sigma) = |\mu - \half| = a(b)/2$. If $a(b) \geq \half$, then $a(b)/2 > |a(b) - \half|$, so that 
$d(F_b F_c a,\sigma) > d(F_b a,\sigma)$,  and hence 
$H(F_bF_c a) < H(F_ba)$ by Lemma {\bf{\ref{monotonic}}}.
If $a(b) < \half$, the replacement of  $F_c$ by $F_e$ replaces $b$ with $b\pr$, and since $a(b\pr) > \half$, the same conclusion
is obtained.
\end{proof}
\vsk
 Although the  metric $d(x,y)$ as defined by (\ref{dmetric}) seems to require a knowledge of $x(z)$ and $y(z)$ for all $z$, we shall see 
 as a by-product of (\ref{differenceequation}) 
 that it is determined by $x(y)$ alone, and the form of that
dependence  is inconsistent with local hidden variables:
\begin{proposition}
\begin{equation}\label{dtop}
d(a,b) = \sqrt{1 - a(b)}  \; \text{ for all } a,b \in \sys .
\end{equation}
\end{proposition}
\begin{proof} 
We can assume  $a(b) \neq 0 $ or $1$,  since otherwise the assertion follows from the definition (\ref{dmetric}) of $d$.  
By proposition {\bf {\ref{uniqueantipode}}},  $b$ has a unique antipode $b\pr$ in ${\cal P}_{ab}$. 
Take the supremum over $z$ of the absolute value on both sides of (\ref{differenceequation}) to obtain:
\begin{equation}
d(a,b) = |2\gamma - 1|d(e,e\pr) =  |2\gamma - 1| = \sqrt{a(b\pr)} = \sqrt{1 - a(b)}.
\end{equation}
\end{proof}
\begin{proposition}
If the property (\ref{tvnproperty1}) holds for $\sys$, its $p$-table cannot be reproduced by a local hidden variable model unless
it is classical i.e.\ unless $p$ only assumes the values zero and unity.
\end{proposition}
\begin{proof}
It is known \cite{FIV1}  (see proof  in  the Appendix) that the relation between $d(a,b)$ and $a(b)$ in any hidden-variable model 
differs from (\ref{dtop}) in that the right side is $1 - a(b)$ rather than $\sqrt{1 - a(b)}$. Hence such models are inconsistent
with Axiom IV unless $a(b)$ only assumes the values $0$ and $1$ for all elements, i.e.\ the system is classical.
\end{proof}
\begin{comment}
Without the square root the triangle inequality for $d$ becomes 
\begin{equation}
x(z) + z(y) \leq 1 + x(y),
\end{equation}
which is Bell's inequality.
It is also shown in the  Appendix  that the  CRQM gives (\ref{dtop}). This strongly suggests that property (\ref{tvnproperty1}) has
brought us close to the CRQM.
\end{comment}
\vsk
\section{Structure of two dimensional subspaces}
We shall now exploit Axioms I-IV to show that two dimensional subspaces are isometric to spheres. It is important to 
distinguish the dimension $N = 2$ of the subspace from the dimension of the spheres which, as we shall see, is determined
by the ``rank" of the subspace given by the following:
\begin{definition}
The rank $R$ of a two dimensional subspace is one less than the maximum number
of elements in a mutually equatorial set.
\end{definition}
\begin{proposition}\label{maintheorem1} If ${\cal P}$ is a two dimensional subspace of rank $R$,
there is a one-one correspondence  $x \to X$ between the elements $x$ of $ {\cal P}$ and the points $X$ of a unit $R$-sphere
such that
\begin{equation}\label{bornruleprelim}
x(y) = \cos^2\half XY
\end{equation}
where $XY$ is the angle subtended at the center of the sphere by the arc joining $X$ and $Y$.
\end{proposition}
\begin{proof}
We shall need several lemmas:
\vsk
Define $\theta(x,y)$ by 
\begin{equation}\label{ptotheta}
x(y) = \cos^2\half(\theta(x,y).
\end{equation}
\begin{definition}
A set of elements $x_1,x_2,,\cdots$ is said to be  {\bf properly mapped} to the points $X_1,X_2,\cdots$ of an $R$-sphere if
$\theta(x_i,x_j) = X_i X_j$ for all $i,j$.
\end{definition}
\noindent
To prove Proposition {\bf{\ref{maintheorem1}}} we shall prove that there is a one-one proper mapping of  the elements of ${\cal P}$ to
the points of an $R$-sphere.
\vsk
Because we are restricting the elements to a two dimensional subspace, we can use the special form (\ref{tvnproperty2}) of Axiom IV. Rewriting this
in terms of $\theta$ defined by (\ref{ptotheta}),  it becomes:
 \begin{equation}\label{tvnproperty3}
 \cos \theta(a,z) + \cos \theta(b,z) = 2\cos\half \theta(a,b)\cos\theta(c,z),
 \end{equation}
 and 
 \begin{equation}\label{solutiontheta}
 \theta(a,c) = \theta(c,b) = \half \theta(a,b).
 \end{equation}
\vsk
 The geometric significance of (\ref{tvnproperty3}) and  (\ref{solutiontheta}) is revealed by the following 
 lemma which applies to an $R$-sphere for any $R$.
 \vsk
 \begin{lemma}\label{sphericaltrig} 
 Let $A,B,C,Z$ be points on an $R$-sphere,  with $C$  the midpoint of the arc joining $A$ and $B$. Then
 \begin{equation}\label{m-sphere-equation}
\cos(AZ) + \cos(BZ) =   2 \cos (\half AB)) \cos(CZ).
\end{equation}
\end{lemma}
\vsk
\begin{proof}
If $A,B$ are antipodes, the right and left sides both vanish.  Assume  they are not antipodes.
 Let $\hat{\bf x}_A, \hat{\bf x}_B, \hat{\bf x}_Z$ be unit vectors from the center to the points $A,B,Z$
respectively of a unit $R$-sphere. The vector ${\bf x} = \half(\hat{\bf x}_A + \hat{\bf x}_B)$ connects the center to the midpoint of the chord
joining $A,B$. Its length is $$|{\bf x}| = \sqrt{(1 +  \hat{\bf x}_A \cdot \hat{\bf x}_B)/2} = \cos (\half AB)$$ which does not vanish if $A$ and $B$ are not antipodes.  Hence $\hat{\bf x}_C = ((\hat{\bf x}_A + \hat{\bf x}_B)/(2 \cos (\half AB))$ is a unit  vector from the center to the midpoint $C$ of the great circle arc 
joining $A$ and $B$.  Hence if  $CZ$ is the arc joining $C$ and $Z$ we have
$$
\cos(CZ) = \hat{\bf x}_C\cdot \hat{\bf x}_Z = (\hat{\bf x}_A\cdot \hat{\bf x}_Z + \hat{\bf x}_B\cdot \hat{\bf x}_Z)/(2 \cos (\half AB)) = $$ $$
(\cos(AZ) + \cos(BZ))/(2 \cos (\half AB))$$
whence
(\ref{m-sphere-equation}) follows. 
\end{proof}
\vsk
 \begin{lemma}\label{antipodeadjunction}
Suppose that a subset of the two-dimensional subspace ${\cal P}$  consisting of  $z_1, z_2, \cdots$  is  properly mapped to the set of points $Z_1,Z_2,\cdots$ on the $R$-sphere $S_R$.  Then the subset consisting
of $z_1,z_2,\cdots,z_1\pr,z_2\pr,\cdots$, where $z_j\pr$ is the antipode of $z_j$ on ${\cal P}$,  is properly mapped to $Z_1,Z_2,\cdots, Z_1\pr,Z_2\pr,\cdots$, where $Z_j\pr$ and $Z_j$ are antipodes on ${\cal S}_R$.
\end{lemma}
\begin{proof}
For any elements $x,y$ of the set $z_1,z_2,\cdots,z_1\pr,z_2\pr,\cdots$, we have
$$x\pr(y) = 1 - x(y) = \sin^2(\theta(x,y)/2) = \cos^2((\theta(x,y) + \pi)/2)= $$
$$\cos^2((XY + \pi)/2) = \cos^2(X\pr Y/2)$$ whence
$$\theta(x\pr,y) = X\pr Y .$$
\end{proof}
\noindent
We refer to the result of this lemma as {\em antipode adjunction} to a properly mapped subset.
\vsk
\begin{lemma}\label{midpointadjunction}
Suppose that a subset of the two-dimensional subspace ${\cal P}$  consisting of  $a \ncl b$ and 
$z_1, z_2, \cdots$  can be  properly mapped to the points $\{A,B,Z_1,Z_2,\cdots \}$ on $S_R$. Let 
$C$ be the midpoint of the shorter great circle arc
joining $A,B$, and let $c = c(a,b)$ of (\ref{tvnproperty1}).
 Then the mapping $c \to C$ extends the proper mapping to the set
$\{a,b,c,z_1,z_2,\cdots\}$.
\end{lemma} \vsk
\begin{proof}
\vsk
Let $z$ be any of the $z_j$'s.
If $\{a,b,z\} \to \{ A, B, Z\}$ is a proper mapping, it follows from (\ref{tvnproperty2}) that
\begin{equation}\label{ifproper}
 \cos(AZ) + \cos(BZ) = 2\cos(\half AB)\cos(\theta(c,z)).
 \end{equation}
 If $a \ncl b$, then $AB \neq \pi$ so that $\cos(\half AB) \neq 0$, whence from
 (\ref{tvnproperty2}) we  have
\begin{equation}
\theta(c,z) = CZ,
\end{equation}
and from  (\ref{solutiontheta}),
\begin{equation}
 \theta(a,c) = \half\theta(a,b) = \half AB = AC,\;\;
  \theta(c,b) = \half\theta(a,b) = \half AB = CB.
  \end{equation}
  Hence $c \to C$  extends the proper mapping  $\{a,b,z\} \to \{A,B,Z\}$ to $\{c,a,b,z\} \to \{C,A,B,Z\}. $
  \end{proof}
  \vsk
  We call the result of this lemma {\em midpoint adjunction} to a properly mapped subset. 
\vsk
\begin{lemma}\label{circleadjunction}
\vsk
Suppose that a subset  ${\cal K}$  of the two-dimensional subspace ${\cal P}$ can be properly mapped to a unit $R$-sphere $S_R$. If $a \ncl b \in {\cal K}$, there exists a subset ${\cal C}_{ab}$ of ${\cal P}$ containing $a$ and $b$, referred to as a {\em circular subset}, which can be properly mapped to a great circle $\tilde{\cal C}_{AB}$ on $S_R$ in such a way that the union of $ {\cal C}_{ab}$ and ${\cal K}$ is properly mapped to $S_R$.
If $a$ and $b$ belong to the equator 
${\cal E}(z)$
 of some state $z \in {\cal P}$, 
then every element $y \in {\cal C}_{ab}$ also belongs to ${\cal E}(z).$
\end{lemma}
\vsk
\begin{proof}  
Let ${\cal P}$ be a two-dimensional subspace with elements 
 $a\ncl b$ .
As shown in Figure 1: Let $a$ and $b$ be mapped to points $A,B$ of a unit  great circle $\tilcal{ C}_{AB}$ on an $R$-sphere separated by an arc
$ \theta(a,b)$. By definition the set $\{a,b\}$ is properly mapped. By antipode adjunction the set $\{a,b,a\pr,b\pr\}$ is properly mapped to $A,B,A\pr,B\pr$.  Let $e_1 = c(a,b)$ and $e_2 = c(a,b\pr)$ be mapped to the midpoints $E_1,E_2$ of the arcs joining $A,B$ and $A,B\pr$ respectively and their 
 antipodes to the points opposite. As noted earlier, $e_1$ and $e_2$ are an equatorial pair.
  By the midpoint adjunction the set $\{a,b,e_1,e_2, a\pr,b\pr ,e_1\pr,e_2\pr\}$ is now properly mapped to $\tilcal{C}_{AB}$. \vsk
  {\centering{\hskip2.0in\parbox{2cm}{\includegraphics[width= 2in]{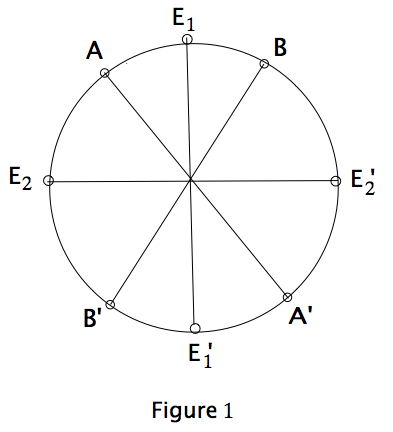} }}}
  \vsk
  The four points $E_1,E_2,E_1\pr,E_2\pr$ are at points that are equally spaced by $\pi/2$.  We can  apply midpoint adjunction $k$ times to extend the proper mapping to include elements of ${\cal P}$ whose images are spaced by $\pi/2^k$. With $k$ arbitrarily large the images become dense on $\tilde{\cal C}_{AB}$ in the round-metric.  The completion of this set of elements in the $d$-metric is the subset ${\cal C}_{ab}$ of ${\cal P}$. From proposition {\bf \ref{equatorialtheorem}} if $a \ncl b \in {\cal E}(z)$, then $e_1,e_2$ and all successive elements of ${\cal C}_{ab}$ obtained by
 midpoint adjunction will be in ${\cal E}(z)$.
  \end{proof}
 \vsk
  Note:  Here we  used Axiom II to insure the completeness of  ${\cal P}$ in the $d$-metric.
  \vsk
    \begin{corollary}\label{equatorcorollary}
  If $x \in {\cal C}_{ab}$  there is a unique pair of  states $y,y\pr \in {\cal C}_{ab} $ such that $x \smile y$ and $ x \smile y\pr$.
   \end{corollary}   
  \begin{proof}
  The required states $\{y,y\pr\}$  are the points of ${\cal C}_{ab}$ with images $\{Y,Y\pr\}$ at the end points of the diameter of  $\tilcal{C}_{AB}$ at right angles to 
  the diameter joining the images 
  $\{X,X\pr\}$ of $\{x,x\pr\}.$
  \end{proof}  
    \vsk
We now combine the results of the above lemmas to complete the proof of proposition
{\bf  \ref{maintheorem1}}:
\vsk 
Let ${\cal P}$ be of rank $R$ so that there are $n = R+1$ elements in a maximal set of mutually equatorial elements. Let
  $e_1, e_2, \cdots, e_n$ of ${\cal P}$ be any such set.  It can be properly mapped to an
  to a unit $R$ sphere  $S_R$ by placing the images  $E_j$ of $e_j$  
for $j = 1,\cdots,n$ at the points of $S_R$ with $i$'th cartesian coordinate $\delta_{ij}$. By Lemma {\bf \ref{circleadjunction}} the proper mapping can be extended to include a subset $ {\cal P}^*$ of ${\cal P}$ with the property that it contains ${\cal C}_{xy}$ for every
pair $x \ncl y$ of its elements. We shall prove that ${\cal P}^* = { {\cal P}}$ and
that the set of images of $ {\cal P}$ covers $S_R$.
\vsk
Suppose first that there is some $x \in {\cal P}$  such that $x \not \in  {\cal P}^*$.
Then the circular set ${\cal C}_{e_1 x}$ can contain no point other than $e_1$ of  
$ {\cal P}^*$ for otherwise $x$ would be in ${\cal P}^*$.  Hence, by the Corollary to Lemma {\bf \ref{circleadjunction}}, there is an   element $f_1 \not \in{\cal P}^*$ but $f_1 \in {\cal C}_{e_1 x}$ such that $ f_1 \smile  e_1$.  Now the circular set ${\cal C}_{e_2 f_1}$ may contain no  element of ${\cal P}^* $ other than $e_2$ for otherwise $f_1$ would belong to ${\cal P}^* $.  Since $f_1 \smile e_1$ and $e_2 \smile e_1$ it follows from Lemma {\bf \ref{circleadjunction}} that every element $f$ of ${\cal C}_{e_2 f_1}$ satisfies $f \smile e_1$.  In particular this is true of an element $f_2$ of ${\cal C}_{e_2 f_1}$, which  exists by the last part of Lemma {\bf \ref{circleadjunction}}, that satisfies $f_2 \smile e_2$. Thus $e_1,e_2,f_2$ is a mutually equatorial set in which $f_2$ is not   an element of ${\cal P}^*$ since otherwise $f_1$ would be a member.  
Next we  construct  ${\cal C}_{e_3 f_2}$ and repeat this process until we produce
a state $f_n$  which belongs to ${\cal P}$ but does not belong to ${\cal P}^*$ and has the property that  $e_1,e_2,\cdots,e_n,f_n$ is a mutually equatorial set.  But this makes
the rank of ${\cal P}$ larger than $R$ which is a contradiction. This proves
that ${\cal P} = {\cal P}^*$.\vsk
Proof  that  $S_R$ is completely covered by the mapping is obtained by essentially the same argument:  Suppose that some point $X$ on the sphere
does not appear. Then no point $Y$ other than $E_1$ on the great circle $C_{E_1X}$ can appear. For if $Y$ were the image of some $y \in {\cal P}$, then the entire circle would appear as the image of the circular set ${\cal C}_{e_1 y}$.
In particular there is a point $F_1$ on  $C_{E_1 X}$ which is equatorial to $E_1$ and does not occur in the mapping. Similarly no point on
 $C_{E_2F_1}$ can occur, and every point on it is equatorial to $E_1$. In particular it contains a point $F_2$ which is equatorial to both $E_1$ and $E_2$. 
  Proceeding in this way we construct a point $F_n$ on $S_R$ which is equatorial to $E_1,\cdots,E_n$ which is impossible. 
 \vsk
 With this result we have established that there is a one-one proper mapping of the completion of ${\cal P}$ of rank $R$ to
 the points of an $R$-sphere which proves proposition {{\bf \ref{maintheorem1}}}. 
 \end{proof}
 \section{Two dimensional subspaces of rank $R = 2$}
 We next make the connection between two dimensional subspaces with rank $R = 2$ and the CRQM for $N$=2 subspaces.
 \vsk
 \begin{proposition}\label{maintheorem2}
 There is a one-one correpondence between the points $x$ of a two dimensional subspace of ${\cal S}$ of rank $R = 2$ and the  points $\bar{x}$ of $CP^1$ such that
\begin{equation}\label{mainequation2}
x_1(x_2) =|\hat{x}_1^*\cdot \hat{ x}_2|^2 \text{  where  } \hat{x} = \bar{x}/|x|
\end{equation}
which is the Born Rule.
\end{proposition}
\begin{proof}
 The space $\cfield P^1$ is the linear space over the complex numbers with elements represented by two {\em homogeneous coordinates} i.e.\ pairs 
$\bar{x}  = \{  \alpha_1,\alpha_2,\},$ not both of which are zero, such that two pairs are identified if their elements  differ by a common non-zero complex factor.  
\vsk
There is a one-one correspondence $X \leftrightarrow {\bar x}$ between the points of a unit 2-sphere and the points of  $\cfield P^1$ obtained by stereographic projection, i.e.\ the correspondence
 $X(\theta,\phi) \leftrightarrow \hat{x} = \bar{x}/|x|$ between a point on a unit 2-sphere with zenith $\theta$ and
 azimuth $\phi$  and the point  of $\cfield P^1$ defined by
 \begin{equation}
\hat{x} =  (\cos\half\theta,e^{i\phi}\sin\half\theta).
\end{equation}
Let $\hat{ x}_1$ and $\hat{ x}_2$ correspond to $X_1 = X(\theta_1,\phi_1)$ and $X_2 = X(\theta_2,\phi_2) $, and let
$\xi_1$ and $\xi_2$ be the cartesian components of unit vectors from the center to $X_1$ and $X_2$ respectively. Then
\begin{equation}
|\hat{x}_1^*\cdot \hat{x}_2|^2 = \half(1 + \xi_1 \cdot \xi_2) = \cos^2(\half X_1X_2),
\end{equation}
where $X_1X_2$ is the great circle arc between points $X_1$ and $X_2$.
 \vsk
The assertion then follows from (\ref{bornruleprelim}).
\end{proof}
\begin{comment}
A  similar result is obtained for ranks $R = 1$ and $R = 4$.  In the former case one maps the $1$-sphere to the real projective line $\rfield P^1$, and in the latter one maps the $4$-sphere to the quaternionic projective line $\qfield P^1$. 
\end{comment} 
 \section{Picking rank $R = 2$}\label{pickingout} 
 To complete our axiomatic system we must choose a property of two dimensional quantum mechanical systems which holds only
 for rank $R = 2$. One such property pointed out in \cite{FIV2} is that the group of transformations leaving a basis
 fixed is a continuous, abelian group. However, since our focus in this paper is on information theory, we shall take a different
 approach based on a property noticed by Sykora\cite{SYK} and Wootters\cite{WOOT}:
 \vsk
 Let $x$ be a state at the north pole of the $R$-sphere.  A measurement in the frame $F_{\bf y} = \{y,y\pr\}$
 transforms $x$ into the mixture $ \rho = p y + (1 - p)y\pr$ with $p = \cos^2\half\theta$ where $\theta$ is the
 co-latitude of $y$. We obtain all possible frames by letting $y$ vary over the upper hemisphere. A simple
 measure of the purity of $\rho$ is  $\delta = |2 p - 1| = \cos\theta$ which is twice its distance  from the maximally mixed state in the $d$-metric
 and ranges between $0$ for the maximally mixed state and $1$ for a pure state. 
   The average of any function of the purity $g(\delta)$  (such as the Shannon entropy)  over all frames   will be given by
 \begin{equation}\label{average}
 \overline{g} = { \int_0^{\pi/2}{g(\cos\theta)(\sin\theta)^{R - 1} d\theta} \over
  \int_0^{\pi/2}{(\sin\theta)^{R - 1} d\theta}} =  { \int_0^{1}{g(\delta)(1 - \delta^2)^{(R - 2)/2} d\delta} \over
 \int_0^{1}{(1 - \delta^2)^{(R - 2)/2} d\delta}}.
  \end{equation}
  For $R = 2$ the right side reduces to $\int_0^{1} g(\delta)d\delta$. This is quite remarkable, for what it says is that
  for $R = 2$ and {\em only} for $R = 2$ the purity $\delta$ is {\em uniformly} distributed on the interval $[0,1]$. 
  For $R = 1$ the negative power of $(1 - \delta^2)$ means that the distribution is weighted towards the pure state value $\delta = 1$
  whereas for $R > 2$ the positive power of $(1 - \delta^2)$ means that  it is weighted towards the maximally mixed state $\delta = 0$.
  \vsk
   This suggests the following choice for our last 
  axiom which picks out $R = 2$:
  \vskip.2in
  \noindent
   {\bf Axiom V: }
   {\em The purities of the mixed states resulting from  random measurements of
  a pure qubit state are uniformly distributed.
   }
  \vskip.2in
  \noindent
 As  noted above,  the appropriate measure of the removal of uncertainty when a measurement is made is the
 information entropy which
takes account of the lack of pre-knowledge of what measurement is going to be made and requires that we average
 the Shannon entropy over all possible measurements. The effect 
of the weighting described above in the real case ($R =1$) and quaternion case ($R = 4$) is that the information entropy will be smaller than in the complex case for real quantum mechanics and larger for quaternionic quantum mechanics.   We verify this by letting $g(\delta)$ be the Shannon entropy
 \begin{equation}
 g(\delta) = - (\half + \delta)\ln(\half + \delta) - (\half - \delta)\ln(\half - \delta)
 \end{equation}
 and obtain the values $2 \ln2 - 1<  1/2 < 7/12$ for $R = 1,2,4$ respectively.
   \vsk
 Given this result we might have considered replacing Axiom V by  elevating to axiomatic status the assertion that the information entropy is $1/2$ for a pure qubit state.
 While this ``works" to pick out $R = 2$, the fact that the value is $1/2$ depends  on the choice of the
 e-base of logarithms in defining the entropy and has no instrinsic physical significance.${\bf ^\dagger}$  In fact one  could just as well have chosen the average of any function of the purity $\delta$  that differs  for different values of $R$.
   Axiom V as stated avoids this .
   \vsk
   \begin{comment}
   The property of a two dimensional complex Hilbert space expressed by Axiom V is a special case of a
   property of all finite dimensional complex Hilbert spaces  proved by Sykora in the appendix to \cite{SYK}, namely that 
    if  $F_{\bf y} x = \sum_{j = 1}^N p_j y_j$ denotes a mixed state
    resulting from a random measurement of a pure state $x$  in a complex Hilbert space of dimension $N$,
   then $p = \{p_1,\cdots,p_N\}$ is uniformly distributed on the ``probability simplex"  $\sum_{j=1}^N p_j = 1$.
  Sykora further remarks that a different result is obtained for a real Hilbert space. Since we shall deduce the
  CMQM for $N > 2$ below from the $N = 2$ case, we only needed the fact that Axioms I-IV implied the
  $R$-sphere structure for $N = 2$ where the points of the probability simplex are linearly related to the purity $\delta$.
 Thus  we could easily deduce the formula  (\ref{average}) and from it see the uniqueness of the complex
 case in producing a uniform distribution. 
 \end{comment}
     \section{Extension to $N > 2 $}
It has now been shown that Axioms I-V imply that the qbit subspaces of $\sys$ are $\cfield P^1$ spaces 
 in which the Born Rule holds. In this section it will be shown that this extends to all of $\sys$ to produce the CRQM.
\vsk
   To carry this out
 we must show that the $\cfield P^1$ coordinatization established for qbits enables us to establish a $\cfield P^{N-1}$ coordinatization for $N$ dimensional spaces consistent with the Born Rule. To accomplish this we first show that  Axioms I-V  imply that $\sys$ is a
 projective geometry.  While it is intuitively clear that if the two dimensional subspaces of an $N$-dimensional projective space are $\cfield P^1$ spaces the space itself must be a $\cfield P^{N-1}$ space,  it is not obvious that a single coordinatization of the space can be carried out in such a way that the Born Rule holds simultaneously  in every two dimensional subspace. This will be verified by first
 showing that  there is a coordinatization such that $x(y) = |\hat{x}^*\cdot\hat{y}|^2$ when at least one of the two elements $x,y$
 lies on an axis and then showing that it extends to arbitrary pairs. 
 \vsk
 \vsk
 --------------------------------------------------------------------
 
 \noindent
 $\dagger$ I am indebted to W. Wootters for this observation.
 \eject

    \begin{lemma}\label{projectivegeomety}
       $ {\cal S}$ is a projective geometry in which the elements  are points,
the two-dimensional subspaces are  lines, and the three dimensional subspaces are planes.
 \end{lemma}
  \begin{proof}
 We must show that:\vsk
  (I) Each pair of distinct points are on exactly one line;   (II) The Veblen-Young axiom holds (see statement below).
  \vsk
 By proposition {\bf \ref{uniquesubspace}}  two distinct points $a,b$ determine exactly one
two-dimensional subspace ${\cal P}_{ab}$.  We must therefore prove (II) which states that
 if $a,b,c,d$ are four  points
 no three of which are colinear, and if the lines ${\cal P}_{ab}$ and ${\cal P}_{cd}$
 intersect, then the lines ${\cal P}_{ac}$ and ${\cal P}_{bd}$ intersect. This is an
 efficient way of asserting that every pair of lines in the plane determined by a pair of intersecting lines will intersect.
\vsk
 We first prove that if  ${\cal P}_{ab}$ and ${\cal P}_{cd}$ intersect they lie in a three-dimensional subspace:
 \vsk
{\centering{\hskip1.0in\parbox{3cm}{\includegraphics[width= 2.75in]{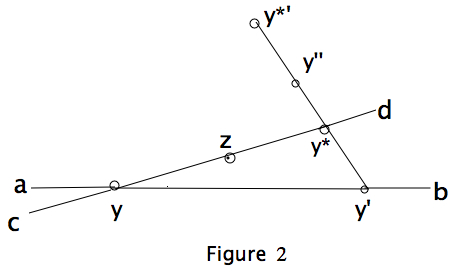} }}}
\vsk
 Refer to Figure 2. 
 Let ${\cal P}_{ab}$ and ${\cal P}_{cd}$ intersect at $y$. Let  $y^\prime$ be 
 its antipode in ${\cal P}_{ab}$, and let $y^*$ be its antipode in ${\cal P}_{cd}$.
 Let $y^{\prime \prime}$ be the antipode of $y^\prime$ in ${\cal P}_{y^\prime y^*}$
  and let $y^{* \prime}$ be the antipode of $y^*$ in ${\cal P}_{y^\prime y^*}$
 Since $y \perp y^\prime$ and $y \perp y^*$ it follows from proposition {\bf \ref{anyorthog}} that $y$ is orthogonal to both  $y^{\prime\prime}$
 and $y^{*\prime}$.
Hence both $\{y,y^\prime,y^{\prime\prime}\} $ and $\{y,y^*,y^{*\prime}\}$ are orthogonal sets . Now let $z$ be any point
on ${\cal P}_{cd}$.
 We have $p(y,z) + p(y^*,z) = 1$ whence
 $p(y^{*\prime},z) = 0$. But since both $\{y^\prime,y^{\prime\prime}\}$ and $\{y^*,y^{*\prime}\}$ are bases of ${\cal P}_{y^\prime y^*}$ it follows from proposition {\bf \ref{anyframespans}} that $p(y^\prime,z) +p(y^{\prime\prime},z) =
p(y^*,z) + p(y^{*\prime},z)$. Hence $p(y^*,z) = p(y^\prime,z) +p(y^{\prime\prime},z)$ whence
$p(y,z) + p(y^\prime,z) +p(y^{\prime\prime},z) = 1.$  Hence $z$ belongs to the three dimensional subspace spanned by  $ \{y,y^\prime,y^{\prime\prime}\}.$ Since any point on ${\cal P}_{ab}$
also lies in this space the assertion follows.
\vsk
Since ${\cal P}_{ab}$ and ${\cal P}_{cd}$ lie in a three-dimensional subspace,
the lines ${\cal P}_{ac}$ and ${\cal P}_{bd}$ lie in that  subspace, and
we can therefore
 complete the proof of the Veblen-Young axiom by showing  that any two 
 distinct lines in a three dimensional subspace intersect.
 \vskip.2in
{\centering{\hskip1.5in\parbox{3cm}{\includegraphics[width= 1.5in]{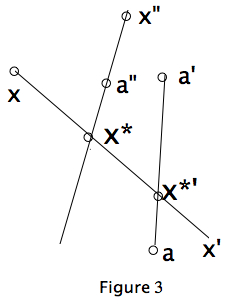} }}}
  \vskip.2in
  \noindent
 See Figure 3. With no loss in generality we can label the lines 
 ${\cal P}_{aa^\prime}$ and  ${\cal P}_{xx^\prime}$. Since $N=3$ there are unique elements $a^{\prime\prime}$ and $x^{\prime\prime}$ such that $a,a^\prime,a^{\prime\prime}$ and
$x,x^\prime,x^{\prime\prime}$ are bases of ${\cal S}$. The line
${\cal P}_{x^{\prime\prime}a^{\prime\prime}}$ contains an element
$x^*$ which is orthogonal to $x^{\prime\prime}$ and hence
lies in ${\cal P}_{xx^\prime}$. Its antipode $x^{*\prime}$ 
is orthogonal to both $x^*$ and $x^{\prime\prime}$ and
hence is orthogonal to ${\cal P}_{x^* x^{\prime\prime}}$ which is identical to  ${\cal P}_{x^{\prime\prime}a^{\prime\prime}}$. Hence it is orthogonal to $a^{\prime\prime}$ so that  $x^{*\prime}$ 
lies in ${\cal P}_{aa^\prime}$ as well as in ${\cal P}_{xx^\prime}$. This establishes the Veblen-Young axiom and completes the proof.
\end{proof}
\vsk
\begin{definition}
Two lines  in a plane $\sys^*$ of $\sys$ are said to be {\em normal}  if the antipode of the intersection in one line
is  orthogonal to its antipode in the other. 
\end{definition}
\noindent
\begin{lemma}\label{nearestpoint}
Let  $z \in {\cal P}$ and $b \in {\cal Q}$  where  ${\cal P}$ and ${\cal Q}$ are normal lines with intersection $x$. 
Then $b(z)$ has a maximum at $b = x$.
\end{lemma}
\begin{proof}
See Figure 4
\vsk
 {\centering\hskip1in{\parbox{2.5cm}{\includegraphics[width=2in]{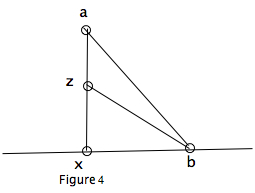} }}  }\vsk
 Recall  equation (\ref{tvnproperty3})
 $$
 \cos \theta(a,z) + \cos \theta(b,z) = 2\cos\half \theta(a,b)\cos\theta(c,z),
$$
which was derived from (\ref{tvnproperty2}) under the assumption $z \in {\cal P}_{ab}$ which is spanned by $c,c\pr$ so that
 $c(z) +
c'(z) = 1.$ Here, however, $z$ is only in ${\cal P}_{ab}$ when $b = x$, but since for arbitrary $z$ we have  $c(z) + c\pr(z) \leq 1$ we can replace the equation  by  the inequality:
 \begin{equation}\label{lessthanoreq}
 \cos\theta(a,z) + \cos \theta(z,b) \leq 2\cos\half \theta(a,b)\cos\theta(c,z)
\end{equation}
with equality if and only if $z \in {\cal P}_{ab}$.
Since ${\cal P}$ and ${\cal Q}$ are normal, $a$ is orthogonal both to $x$ and to the antipode of $x$ on ${\cal P}$ (not shown in the figure) and hence
by proposition {\bf{\ref{anyorthog}}}  is orthogonal to every point of ${\cal P}$. Hence $\theta(a,b) = \pi$ so that the right side of (\ref{lessthanoreq}) is zero. Let $\theta(a,z) + \theta(z,b) =  \pi - \epsilon$ so that (\ref{lessthanoreq}) becomes
\begin{equation}\label{atpi}
\cos\theta(a,z) \leq  \cos( \theta(a,z) + \epsilon)
\end{equation}
with equality if and only if $z \in {\cal P}_{ab}$ which occurs for $b = x$, i.e.\ $\epsilon = 0$. 
For small, non-zero $|\epsilon| $ equation (\ref{atpi}) implies   $\epsilon < 0$,
 so that $\theta(a,z) + \theta(z,b) \geq  \pi $. Hence $\theta(z,b)$ has a minimum  at $b = x$ and hence 
 $b(z) = \cos^2\half\theta(b,z)$ has a maximum.
 \noindent
\end{proof}
\begin{lemma}\label{ppptheorem}
Let $c$ be the intersection of a pair of normal lines ${\cal P}$ and ${\cal Q}$ in a plane $\sys^*$. If $a\in {\cal P}$ and
$z \in {\cal Q}$, then
\begin{equation}\label{pppequation}
a(z) = a(c)c(z) .
\end{equation}
\end{lemma}
\begin{proof}
\vsk
 {\centering\hskip1in{\parbox{2.5cm}{\includegraphics[width=2.5in]{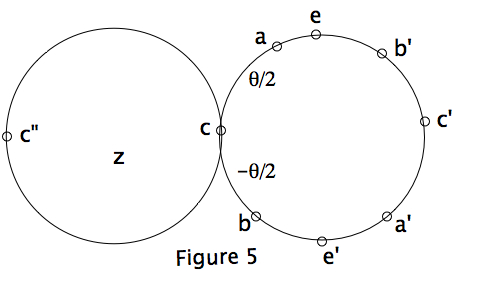} }}  }

 \vsk
Figure 5 shows a pair of two dimensional subspaces ${\cal P}$ and ${\cal Q}$ with intersection $c$  that are normal to one another. More precisely it
shows the images of these subspaces under the BRC mapping to unit $2$-spheres.
The antipodes of  $c$ on ${\cal P}$ and ${\cal Q}$ are  $c'$ and  $c''$ respectively.  $z$ is an arbitrary point of ${\cal Q}$. Since $c\pr$ is orthogonal
both to $c$ and $c''$ it follows from proposition {\bf{\ref{anyorthogonal}}} that $c\pr(z) = 0$.  Hence (\ref{tvnproperty1}) becomes
\begin{equation}\label{pluseq}
a(z) + b(z) = 2 \lambda c(z) \text{ where } \lambda = a(c) = \cos^2\textstyle{1\over 4}\theta.
\end{equation}
where $\theta$ is the angle between the images of $a$ and $b$ which lie on a great circle through $c$.  By varying the angle $\theta$ we can vary $a$ and $b$ on the circle while holding fixed both the midpoint $c$ of $a$ and $b$ and the midpoint $e$ of $a$ and $b\pr$.
Thus (\ref{differenceequation}) can be written
\begin{equation}\label{minuseq}
a(z) - b(z) =  (2\gamma - 1)K(z), \text{ where  }\gamma =  \half(1 + \sqrt{a(b\pr)}) = \half(1 + \sin\half\theta),
\end{equation}
and $$K(z) = e(z) - e\pr(z)$$
is independent of $\theta$.
Adding (\ref{pluseq}) and (\ref{minuseq}):
\begin{equation}\label{sumeq}
a(z) = \half K(z)\sin\half\theta + c(z) \cos^2\textstyle{1\over 4}\theta.
\end{equation}
By Lemma {\bf{\ref{nearestpoint}}} $a(z)$ has a maximum when $a$ is at the intersection point, i.e.\ at  $\theta = 0$. It follows
that $K(z) = 0$ and hence that $a(z) = b(z)$  whence from (\ref{pluseq}) and (\ref{lambdaeq}) we have
\begin{equation}
a(z) = a(c)c(z).
\end{equation}
\end{proof}
 \begin{definition}
 Let $\sys$ have dimension $N$, and let ${\bf a} = \{a_1,\cdots,a_N\}$ be a basis..  The lines $q_{ij} \equiv {\cal P}_{a_i a_j}$ with $i \neq j$ are referred to as the {\em axes} of the ${\bf a} $-basis. The set $\Omega({\bf a})$ is the set of points belonging to all of the axes of the ${\bf a}$-basis.
 \end{definition}
 \begin{lemma}\label{pointsonaxis}
  Let $\sys$ have dimension $N$, and let ${\bf a} = \{a_1,\cdots,a_N\}$ be a basis. Let $\overline{S}$ be a $\cfield P^{N-1}$ space the elements of which are represented by $N$ component homogeneous complex coordinate vectors ${\bar \xi} = (\xi_1,\xi_2,\cdots,\xi_{N })$.
 There is a mapping $z \to \bar{z}$  of  the elements of $\sys$ to $\overline{S}$ such that the Born Rule holds between pairs $z_1,z_2$ provided that at least one member belongs to $\Omega({\bf a})$.
 \end{lemma}
\begin{proof}
By induction: The case $N = 2$  is  Proposition {\bf{\ref{maintheorem2}}}. Assume true for $N  - 1$. 
 Map the basis elements $a_j$ to $\overline{a_j}$ with $j$'th component $1$ and the others $0$. For $j = 1,\cdots,N$ map the $N - 1$ dimensional subspace $\sys^{(j)}$ of $\sys$ which is orthogonal to $a_j$ to the $N -1$ dimensional  subspace $\overline{\sys^{(j)}}$ of $\overline{\sys}$ orthogonal to $\overline{a_j}$.  By hypothesis of the induction these correspondences can be made in such a way that the Born Rule holds between elements provided at least one of them belongs to one of the axes of $\overline{\sys^{(j)}}$.  For each $x \in \sys^{(j)}$ we  map the  line ${\cal P}_{a_j x}$ to the two dimensional subspace of $\overline{\sys}$ spanned by
 $\bar{x}$ and $\overline{a_j}$ which we can do in such a way that the Born Rule holds between its elements.
 \vsk
 {\centering{\hskip1.0in\parbox{3cm}{\includegraphics[width= 2.75in]{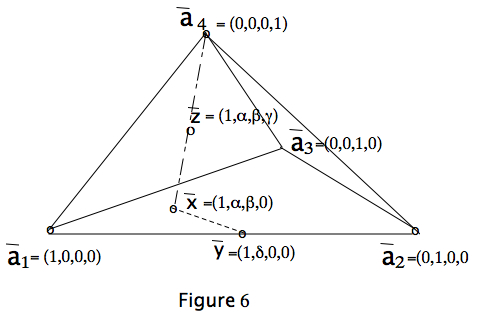} }}}
 \vsk
To simplify the next part of the  proof  let us consider the case $N = 4$ (see Figure 6)  which will make  the general argument clear.
The points of $\sys^{(4)}$ are mapped to $\overline{\sys^{(4)}}$ and so have coordinates of the form $x = (1,\alpha,\beta,0)$.
(Note that with homogeneous coordinates we can represent the elements in this way.  For example $\alpha \to \infty$ corresponds
to the element $(0,1,0,0)$.) If $z$ lies on ${\cal P}_{a_4x}$ it will be mapped to a linear combination of $x$ and $a_4$ and
hence to a point $\bar{z} = (1,\alpha,\beta,\gamma)$ for some $\gamma \in \cfield$. Now let $y$ be an element lying on one of the
axes of $\sys^{(4)}$ say ${\cal P}_{a_1 a_2}$ which is mapped to $\bar{y} = (1,\delta,0,0)$ for some $\delta \in \cfield$.
By these mappings the Born Rule holds for the pair $z,x$ and for the pair $x,y$.  Moreover the lines ${\cal P}_{a_4 x}$ and
${\cal P}_{xy}$ are normal to one another since the antipode of $x$ in ${\cal P}_{xy}$ lies in $\sys^{(4)}$ and hence is orthogonal
to the antipode $a_4$ of $x$ in ${\cal P}_{a_4 x}$.  Hence by lemma {\bf{\ref{ppptheorem}}}
\begin{equation}
z(y) = z(x)x(y) = |\hat{z}^*\cdot \hat{x}|^2  |\hat{x}^*\cdot \hat{y}|^2  = 
 { |1 + \alpha^*\delta|^2 \over (1 + |\alpha|^2 + |\beta|^2 + |\gamma|^2)(1 + |\delta|^2)} = |\hat{z}^*\cdot \hat{y}|^2.
\end{equation}
In similar fashion one establishes that the Born Rule holds between every $z \in \sys$ and any element lying on
one of the axes of $\sys$.
  \end{proof}
  \vsk
  We are now equipped to prove our main theorem: 
  \vsk
  {\bf {Theorem}}
  \vsk
Axioms I-V imply that  there is a mapping  of  the elements of $\sys$ of dimension $N$  to a $\cfield P^{N-1}$ space
 such that the Born Rule holds between every pair of elements.
\vsk
 \begin{proof}
 \vsk
 Choose a basis $a_1,\cdots,a_N$ of $\sys$  and perform the mapping described in {\bf Lemma}  {\bf{\ref{pointsonaxis}}}.
 Let $b_1,\cdots,,b_N$ be another basis of $\sys$. We have
 \begin{equation}
 \sum_{j = 1}^N b_j(s) = 1 \text { for all } s \in \sys^*.
 \end{equation}
 If $s$ lies on any of the axes ${\cal P}_{a_k a_n}$ used for the map above  the Born Rule holds  between $s$ and
 every element and in particular between $s$ and $b_1,\cdots,b_N$.  Hence:
 \begin{equation}
 \sum_{j = 1}^N  |\hat{b_j}^*\cdot \hat{s}|^2 = 1 .
 \end{equation}
Apply this result when $s = a_k$ for $k = 1,\cdots,N$  to obtain
 \begin{equation}
 \sum_{j = 1}^N  |\hat{b_j}^*\cdot \hat{a_k}|^2 = 1 \text{  for  } k = 1,\cdots,N .
 \end{equation}
 Also apply it when $s$  is an arbitrary point on the axis ${\cal P}_{a_k a_n}$ for all
 pairs $k,n$ with $k \neq n$ i.e.\ for  $\hat{s} = \alpha \hat {a}_k + \beta \hat{a}_n$
 where $\alpha$ and $\beta$ are arbitrary. 
For $k = n$ we have
\begin{equation}
\sum_{j=1}^N |\hat{b}_j^*\cdot\hat{a}_{k}|^2  = 1,
\end{equation}
and for $k\neq n$
\begin{equation}
(|\alpha|^2 + |\beta|^2)^{-1}\sum_{j=1}^N |\hat{b}_j^*\cdot(\alpha\hat{a}_k + \beta\hat{a}_n) |^2 = 1
\end{equation}
The second equation works out to
\begin{equation}
{|\alpha|^2\over{|\alpha|^2 + |\beta|^2}}\sum_{j=1}^N |\hat{b}_{j}^*\cdot \hat{a}_k|^2 + 
{|\beta|^2\over{|\alpha|^2 + |\beta|^2}}\sum_{j=1}^N |\hat{b}_{j}^*\cdot \hat{a}_n|^2 + 
\end{equation}
\begin{equation}
2\Re{{{\alpha^* \beta }\over{ |\alpha|^2 + |\beta|^2}}}\sum_{j=1}^N (\hat{a}_{k}^*\cdot \hat{b}_j) 
(\hat{a}_{n}\cdot \hat{b}_j^*) = 1
\end{equation}
and since $\alpha$ and $\beta$ are arbitrary the sum must vanish for $k \neq n$. Hence
\begin{equation}
\sum_{j=1}^N (\hat{a}_{k}^* \cdot \hat{b}_j)
(\hat{a}_{n}\cdot \hat{b}_{j}^*) = \delta_{kn}
\end{equation}
Thus the matrix $M_{kj} = \hat{a}_{k}^*\cdot \hat{b}_j$ is unitary. Hence the coordinates
are transformed by a unitary transformation under a change of basis.  Given any pair of elements
we can choose a basis  such that they lie on an axis and hence such that the Born Rule holds
between them.
But the scalar
product is invariant under unitary transformations whence it follows that the Born Rule remains
valid in every basis.
\end{proof}
\vsk
We have now completed proof that Axioms I-V imply the CRQM.  In the usual Dirac notation  we have established that pure states $x$ in $ \sys$ obeying these axioms correspond to  kets $\ket{x}$ such that the Born Rule $x(y) =|\bra{x}y\rangle|^2$ holds.
\vsk
\section{Discussion}
While our derivation of the CRQM from Axioms I-V has been lengthy, we have obtained the result with no dimensional restriction and
without having to use  Gleason's Theorem which is itself quite lengthy. Moreover our inductive proof of the extension to $N > 2$ has
gotten the result with no more than the rudiments of projective geometry, in particular avoiding the proof that
if the planes of an $N$-dimensional projective space are $\cfield P^2$ the space is $\cfield P^{N-1}$\cite{BUM}. 
\vsk
A key role in the axiomatic system presented here was played by Axiom IV, motivated by the entropic Turing-von Neumann effect in
which the loss of purity of a state resulting from a measurement is reduced by an intermediate measurement.   
It is remarkable
that the impurity can be made arbitrarily small by a sufficiently large number of intermediate measurements. Indeed, one can
reproduce unitary (Schr\"{o}dinger) dynamics with arbitrarily high accuracy in this way. For consider a sequence of
measurements performed on a system in frames containing states $\ket{n}$ which are related by
$\ket{n} = e^{-i\tau H}\ket{n-1}$ for some Hamiltonian $H$ and time interval $\tau$. The probability $p_n$ for finding the system in state $\ket{n}$ after $n$ measurements starting with $\ket{0}$  is at least $|\bra{n}n-1\rangle|^2 p_{n-1}$.  If $\tau\Delta << 1$, where
$\Delta = (\bra{0}H^2\ket{0} - \bra{0}H\ket{0}^2)^{\half}$ is the dispersion of $H$ we have
\begin{equation}
|\bra{n}n-1\rangle|^2 \approx 1 - (\tau \Delta)^2,
\end{equation}
whence
\begin{equation}
1 \geq p_n \geq ( 1 - (\tau \Delta)^2)^n.
\end{equation}
At a  time $T$, when $ n = T / \tau$ measurements have been made, the right side approaches $e^{-\tau T \Delta^2}$.
Thus given any finite time $T$ and dispersion $\Delta$ {\em one can choose a sufficiently small interval $\tau$ between measurements that the initial state $\ket{0}$ is ``coaxed" by the sequence of measurements into a mixed state so dominated by $\ket{n}$  that the entropy increase is arbitrarily small.} (The so-called ``watchdog" effect or ``quantum Zeno effect" occurs when the repeated measurements are made in a frame containing the initial state which supresses the evolution altogether.)
\vsk
Thus, while collapse due to measurement cannot be reproduced by unitary dynamics, a fact that gives rise to the measurement problem, we see that the converse is not true, i.e.\  {\em unitary dynamics can be reproduced to arbitrary accuracy by a sequence of collapses.} It is thus theoretically possible that what appears to us as Schr\"{o}dinger evolution is  a very good approximation to a process in which an interaction Hamiltonian
``guides" a sequence of collapse processes happening during  very small time intervals.  By choosing $\tau$ small enough
 no increase in entropy would be detected even on a cosmic time scale.
\vsk

\section{Appendix  }
For conventient reference we here reproduce the derivation\cite{FIV1,FIV2} of the relationship
between $p$ and the $d$-metric in the conventional model and  local hidden variable
theories.  
\vsk
In a local hidden variable theory one has a set $\Lambda$ with a measure $\mu$ such that
\begin{equation}
p(x,y) = \mu(\Lambda(x)\cap\Lambda(y)),\quad \mu(\Lambda(x)) = 1, \forall x.
\end{equation}
To evaluate the $d$-metric we must compute the supremum over $z$ of.
 $|\mu(\Lambda(x)\cap\Lambda(z)) - \mu(\Lambda(y)\cap\Lambda(z))|$. But we note
that the contribution coming from any overlap of $\Lambda(x)$ and $\Lambda(y)$ will cancel. Hence one
can compute the $z$ maximizing the expression  as if the sets are disjoint. This occurs when
either $z = x$ or $z = y$ and gives $1 - \mu(\Lambda(x)\cap\Lambda(y))$ whence 
\begin{equation}
d(x,y) = 1 - p(x,y)
\end{equation}
which disagrees with (\ref{dtop}).
\vsk 
In the CRQM we have (reverting to Dirac notation):
\begin{equation}
p(x,y) = |\bra{x}y\rangle|^2 =  Tr(\pi(x)\pi(y)),\;\; \pi(z) \equiv \ket{z}\bra{z},
\end{equation}
whence
\begin{equation}
d(x,y) = \sup_z|Tr(\pi(x)\pi(z)) - Tr(\pi(y)\pi(z))| = \sup_z|\langle z|\pi(x) -\pi(y)|z \rangle|.
\end{equation}
But this is just the largest eigenvalue of $\pi(x) - \pi(y)$. Since the $\pi$'s  are
projectors:
  \begin{equation}
(\pi (x) - \pi (y))^3 = (1 - |\langle x|y \rangle|^2)(\pi(x) - \pi(y))
\end{equation}
and one reads off the largest eigenvalue to obtain (\ref{dtop}).
\vsk
\vsk
 \centerline{{\bf {Acknowledgements}}}
 \vsk
I am gratefull to Prof. Philip Goyal and the Perimeter Institute for hosting the
``Reconstructing Quantum Mechanics" conference in 2009 which led me to revisit
my 1994 paper and resulted in significant simplifications in the argument.
Conversations with Prof. Wm. Wootters at the conference  led me to think further
about the problem of eliminating exotic alternatives to conventional quantum mechanics
which together with the paper of Stotland et al. stimulated the approach taken in this paper.

\end{document}